\documentclass{article}

\usepackage{IEEEtrantools}
\usepackage{color}
\usepackage{verbatim}
\usepackage{amsfonts,amstext,mathrsfs,amsmath,amssymb,amsthm}
\usepackage{hyperref}

\newcommand{\ket}[1]{\left| #1 \right>} 
\newcommand{\bra}[1]{\left< #1 \right|} 
\newcommand{\braket}[2]{\left< #1 \vphantom{#2} \right|
 \left. #2 \vphantom{#1} \right>} 
 
\newtheorem{definition}{Definition}
\newtheorem{theorem}{Theorem}
\newtheorem{proposition}{Proposition}
\newtheorem{lemma}[theorem]{Lemma}
\newtheorem{corollary}[theorem]{Corollary}

\title{Tensor Rank and Strong Quantum Nondeterminism in Multiparty Communication}
\author{
	Marcos Villagra\thanks{\texttt{marcos.villagra@acm.org}, Graduate School of Information Science, Nara Institute of Science and Technology, Nara 630-0192, Japan.}
	\and
	Masaki Nakanishi\thanks{\texttt{m-naka@e.yamagata-u.ac.jp}, Faculty of Education, Art and Science, Yamagata University, Yamagata, Japan.}
	\and
	Shigeru Yamashita\thanks{\texttt{ger@cs.ritsumei.ac.jp}, Department of Computer Science, Ritsumeikan University, Shiga 525-8577, Japan.}
	\and
	Yasuhiko Nakashima\thanks{\texttt{nakashim@is.naist.jp}, Graduate School of Information Science, Nara Institute of Science and Technology, Nara 630-0192, Japan.}
}
\date{}

\begin{document}
\maketitle
\begin{abstract}
In this paper we study quantum nondeterminism in multiparty communication. There are three (possibly) different types of nondeterminism in quantum computation: i) strong, ii) weak with classical proofs, and iii) weak with quantum proofs. Here we focus on the first one. A strong quantum nondeterministic protocol accepts a correct input with positive probability, and rejects an incorrect input with probability 1. In this work we relate strong quantum nondeterministic multiparty communication complexity to the rank of the communication tensor in the Number-On-Forehead and Number-In-Hand models. In particular, by extending  the definition proposed by de Wolf to {\it nondeterministic tensor-rank} ($nrank$), we show that for any boolean function $f$ when there is no prior shared entanglement between the players, 1) in the Number-On-Forehead model, the cost is upper-bounded by the logarithm of $nrank(f)$; 2) in the Number-In-Hand model, the cost is lower-bounded by the logarithm of $nrank(f)$. Furthermore, we show that when the number of players is $o(\log\log n)$ we have that $NQP\nsubseteq BQP$ for Number-On-Forehead communication.
\end{abstract}
\noindent{{\bf Keywords:} communication complexity, multiparty communication, quantum computation, quantum nondeterminism, tensor rank

\section{Introduction}
Nondeterminism plays a fundamental role in complexity theory. For instance,  the $P$ vs $NP$ problem asks if nondeterministic polynomial time is strictly more powerful than deterministic polynomial time. Even though nondeterministic models are unrealistic, they can give insights into the power and limitations of realistic models (i.e., deterministic, random, etc.).

There are two ways of defining a nondeterministic machine, using randomness or as a proof system: a nondeterministic machine {\it i}) accepts a correct input with positive probability, and rejects an incorrect input with probability one;  or {\it ii}) is a deterministic machine that receives besides the input, a proof or certificate which exists if and only if  the input is correct. For classical machines (i.e., machines based on classical mechanics), these two notions of nondeterminism are equivalent. However, in the quantum setting they can be different. In fact, these two notions give rise to (possibly) three different kinds of quantum nondeterminism. In {\it strong quantum nondeterminism}, the quantum machine accepts a correct input with positive probability. In {\it weak quantum nondeterminism}, the quantum machine outputs the correct answer when supplied with a correct proof, which could be either classical or quantum.

The study of quantum nondeterminism in the context of query and communication complexities started with de Wolf \cite{DeWolf2000}. In particular, de Wolf \cite{DeWolf2000,DeWolf2003} introduced the notion of {\it nondeterministic rank} of a matrix, which was proved to completely characterize strong quantum nondeterministic communication. In the same piece of work, it was proved that strong quantum nondeterministic protocols are exponentially stronger than classical nondeterministic protocols. Similarly, Le Gall \cite{LeGall2006} studied  weak quantum nondeterministic communication with classical proofs and showed a quadratic separation for a total function.

Weak nondeterminism seems a more suitable definition, mainly due to the requirement of the existence of a proof, a concept that plays fundamental roles in complexity theory. In contrast, strong nondeterminism lends itself to a natural mathematical description in terms of matrix rank. Moreover, strong nondeterminism is a more powerful model capable of simulating weak nondeterminism with classical and quantum proofs. The reverse, if weak nondeterminism is strictly a less powerful model or not is still an open problem.

The previous results by de Wolf \cite{DeWolf2003} and Le Gall \cite{LeGall2006} were on the context of 2-party communication complexity, i.e., there are two players with two inputs $x,y\in\{0,1\}^n$ each, and they want to compute a function $f(x,y)$. Let $rank(f)$ be the rank of the communication matrix $M_f$, where $M_f[x,y]=f(x,y)$. A known result by \cite{Buhrman2001b} is $\lceil\frac{1}{2}\log rank(f)\rceil \leq Q(f)\leq D(f)$, where $D(f)$ is the deterministic communication complexity of $f$ and $Q(f)$ the quantum exact communication complexity\footnote{All logarithms in this paper are base 2.}. It is conjectured that $D(f)=O(\log^c rank)$ for some arbitrary constant $c$. This is the {\it log-rank conjecture} in communication complexity, one of the biggest open problems in the field. If it holds, it will imply that $Q(f)$ and $D(f)$ are polynomially related. This is in stark contrast to the characterization given by de Wolf \cite{DeWolf2003} in terms of the nondeterministic matrix-rank, which is defined as  the minimal rank of a matrix (over the complex field) whose $(x,y)$-entry is non-zero if and only if $f(x,y)=1$.

\subsection{Contributions}
In this paper, we continue with the study of strong quantum nondeterminism in the context of multiparty protocols. Let $k\geq 2$ be the number of players evaluating a function $f(x_1,\dots,x_k)$ where each $x_i\in\{0,1\}^n$. The players take turns predefined at the beginning of the protocol. Each time a player sends a bit (or qubit if it is a quantum protocol), he sends it to the player who follows next. The computation of the protocol ends when the last player computes $f$. The communication complexity of the protocol is defined as the minimum number of bits that need to be transmitted by the players in order to compute $f(x_1,\dots,x_k)$. There are two common ways of communication: The Number-On-Forehead model (NOF), where player $i$ knows all inputs except $x_i$; and, Number-In-Hand model (NIH), where player $i$ only knows $x_i$. Also, any protocol naturally defines a {\it communication tensor} $T_f$, where $T_f[x_1,\dots,x_k]=f(x_1,\dots,x_k)$.

Tensors are natural generalizations of matrices. They are defined as multi-dimensional arrays while matrices are 2-dimensional arrays. In the same way, the concept of matrix rank extends to {\it tensor rank}. However, the nice properties of matrix rank do not hold  anymore for tensors; for instance, the rank could be different if the same tensor is defined over different fields; see the survey paper by Kolda and Bader \cite{Kolda2009}.

We extend the concept of nondeterministic matrices to {\it nondeterministic tensors}. The {\it nondeterministic tensor rank}, denoted $nrank(f)$, is the minimal rank of a tensor (over the complex field) whose $(x_1,\dots,x_k)$-entry is non-zero if and only if $f(x_1,\dots,x_k)=1$.

Let $NQ_k^{NOF}$ and $NQ_k^{NIH}$ denote the $k$-party strong quantum nondeterministic communication complexity without prior shared entanglement for the NOF and NIH models respectively. 

\begin{theorem}\label{the:bounds}
Let $f:(\{0,1\}^n)^k\to \{0,1\}$, then $NQ_k^{NOF}(f)\leq\lceil \log nrank(f)\rceil+1$,  and $NQ_k^{NIH}(f) \geq \lceil \log nrank(f)\rceil+1$.
\end{theorem}

This theorem generalizes previous results by de Wolf \cite{DeWolf2003}. Also, since $NQ_k^{NIH}$ is a lower bound for exact NIH quantum communication\footnote{An exact quantum protocol accepts a correct input and rejects an incorrect input with probability 1.}, denoted $Q_k^{NIH}$, we obtain the following corollary:

\begin{corollary}\label{cor:exact}
$\lceil \log nrank(f)\rceil+1 \leq Q_k^{NIH}(f)$.
\end{corollary}

The proof of Theorem \ref{the:bounds} is given in Section \ref{sec:bound}. Even though it is a generalization of the techniques of \cite{DeWolf2003}, it requires technical insight. The proof does not generalize in an straightforward manner and it does not yield the same characterization as in the 2-player case. For example, $NQ_k^{NOF}$ cannot be lower-bounded in general by the tensor rank. To see this consider the $k$-party equality function $EQ$ given by $EQ_k(x_1,\dots,x_k)=1$ if and only if $x_1=\cdots=x_k$.  A nondeterministic tensor for $EQ_k$ is superdiagonal\footnote{An order-$k$ tensor is \emph{superdiagonal} when $T[x_1,\dots,x_k]\neq 0$ if and only if $x_1=\cdots x_k$.} with non-zero entries in the main diagonal, and 0 anywhere else. Thus, it has $2^n$ rank, and implies by Theorem \ref{the:bounds} that $NQ_k^{NOF}(EQ_k)\leq n+1$ and $NQ_k^{NIH}(EQ_k)\geq n+1$. In particular, the communication complexity of $EQ_k$ is upper-bounded by $\mathcal{O}(n)$ in the NOF model. However, it is easy to show that in the NOF model there exists a classical protocol for $EQ_k$ with a cost of 2 bits\footnote{In the \emph{blackboard model} (explained in Section \ref{sec:preliminaries}) for $k\geq 3$, let the first player check if $x_2,\dots,x_k$ are equal. If they are, he sends a 1 bit to the second player, who will check if $x_1,x_3,\dots,x_k$ are equal. If his strings are equal and he received a 1 bit from the first player, he sends a 1 bit to all players indicating that all strings are equal. In the \emph{message-passing} model the same protocol has a cost of $O(k)$ bits.}. Hence, the characterization for the 2-player case does not extends to the multiplayer case. In contrast, the lower bound on $NQ_k^{NIH}(EQ_k)$ that follows from Theorem \ref{the:bounds} is not that loose; using the trivial protocol, where all players send their inputs, we have that $NQ_k^{NIH}(EQ_k)=\mathcal{O}(kn)$. Thus, Theorem \ref{the:bounds} yields a tight bound for $EQ_k$ whenever $k=O(1)$. However, whether the same phenomenon extends to all functions in the NIH model is unknown. See below in this section for some consequences on constructing tensors with high rank.

A more interesting function is the generalized inner product  $GIP_k(x_1,\dots,x_k)=(\sum_{i=1}^k \bigwedge_{j=1}^n x_{ij}) \mod 2$. We know that $nrank(GIP_k)\geq (k-1)2^{n-1}+1$ (see Section \ref{app:gip-rank} for a proof). Thus, we have the following result.
\begin{proposition}
$NQ_k^{NIH}(GIP_k)\geq n+\lceil \log(k-1)\rceil$.
\end{proposition}
In NIH, using the trivial protocol, we obtain (with Corollary \ref{cor:exact}) a bound in quantum exact communication of $ n+\lceil \log(k-1)\rceil-1 \leq Q_k^{NIH}(GIP_k) \leq (k-1) n+1$. Improving the lower bound will require new techniques for explicit construction of linear-rank tensors with important consequences to circuit lower bounds; see for example Raz \cite{Raz2010} and the paper by Alexeev, Forbes and Tsimerman\cite{Alexeev2011} for state-of-the-art tensor constructions. In general, we are still unable to upper-bound $NQ_k^{NIH}(f)$ in terms of $\log nrank$. This way we have a new \emph{log-rank conjecture} for strong quantum nondeterministic communication complexity.

Although the bounds given by Theorem \ref{the:bounds} could be loose for some functions, they are good enough  for other applications. For instance, we show in Section \ref{sec:separation} a separation between the NOF models of strong quantum nondeterminism and bounded-error quantum communication. We do so by applying Theorem \ref{the:bounds} to a total function explicitly constructed for this task. This result could be considered as the quantum analog of a separation previously proved in \cite{David2009,Chatt2008,Gavinsky2010} between classical nondetermistic and randomized NOF communication.

\section{Preliminaries}\label{sec:preliminaries}
In this paper we assume basic knowledge of communication complexity and quantum computing. We refer the interested reader to the books by \cite{Kushilevitz1997} and \cite{nielsen00} respectively. In this section we give a small review of tensors and quantum communication.
\subsection{Tensors}
A {\it tensor} is a multi-dimensional array defined over some field. An order-$d$ tensor is an element of the tensor product of $d$ vector spaces. 

\begin{definition}[Simple Tensor]
Let $\ket{v_i}\in V^{n_i}$ be an $n_i$-dimensional vector  for $1\leq i \leq d$ on some vector space $V^{n_i}$. The $j_i^{th}$ component of $\ket{v_i}$ is denoted by $v_i(j_i)$ for $1\leq j_i \leq n_i$. The tensor product of $\{\ket{v_i}\}$ is the tensor $T\in V^{n_1}\otimes\cdots\otimes V^{n_d}$ whose  $(j_1,\dots,j_d)$-entry is $v_1(j_1)\cdots v_d(j_d)$, i.e., $T[j_1,\dots,j_d]=v_1(j_1)\cdots v_d(j_d)$. Then $T=\ket{v_1}\otimes\cdots\otimes \ket{v_d}$ and we say $T$ is a rank-1 or simple order-$d$ tensor. We also say that  a tensor is of high order if $d\geq 3$. 
\end{definition}

From now on, we will refer to high-order tensors simply as tensors, and low-order tensor will be matrices, vectors, and scalars as usual.

It is important to note that the set of simple tensors spans the space $V^{n_1}\otimes\cdots\otimes V^{n_d}$, and hence, there exist tensors that are not simple. This leads to the definition of rank.

\begin{definition}[Tensor Rank]
The rank of a tensor is the minimum $r$ such that $T=\sum_{i=1}^r A_i$ for simple tensors $A_i$.
\end{definition}

This agrees with the definition of matrix rank. The complexity of computing tensor rank was studied by H\aa stad \cite{Hastad1990} who showed that it is $NP$-complete for any finite field, and $NP$-hard for the rational numbers.

The process of arranging the elements of an order-$k$ tensor into a matrix is known as {\it matrization}. Since there are many ways of embedding a tensor into a matrix, in general the permutation of columns is not important, as long as the corresponding operations remain consistent; see Kolda and Bader\cite{Kolda2009}.

\subsection{Strong Quantum Nondeterministic Multiparty Communication}
In a multiparty communication protocol there are $k\geq 3$ players trying to compute a function $f$. Let $f:X^k\to \{0,1\}$ be a function on $k$ strings $x=(x_1,\dots, x_{k})$, where each $x_{i}\in X$ and $X=\{0,1\}^n$. There are two common ways of communication between the players: The Number-In-Hand  (NIH) and the Number-On-Forehead (NOF) models. In NIH, player $i$ only knows $x_{i}$, and in NOF, player $i$ knows all inputs except $x_{i}$. First we review the classical definition.
\begin{definition}[Classical Nondeterministic Protocol]
Let $k$ be the number of players. In order to communicate, the players take turns in an order predefined at the beginning of the protocol. Each player sends exactly one bit to the player that follows next. The computation of the protocol ends when the last player computes $f$. If $f(x)=1$ then, the protocol accepts $x$ with positive probability; if $f(x)=0$, the protocol rejects $x$ with probability 1. The cost of the protocol is the total number of bits communicated.
\end{definition}

Hence, the {\it classical nondeterministic multiparty communication complexity}, denoted $N_k(f)$, is defined as the minimum number of bits required to compute $f(x)$.  If the model is NIH or NOF, we add a superscript $N_k^{NIH}(f)$ or $N_k^{NOF}(f)$ respectively. Note that, the definition of the multiparty protocols in this paper (classical and quantum) are by \emph{message-passing}, i.e., a player sends a bit only to the player that follows next. This is in contrast to the more common {\it blackboard model}. In this latter model, when a player sends a bit, he does so by broadcasting it and reaching all players immediately. Clearly, any lower bound on the blackboard model is a lower bound for the message-passing model in this paper.

To model NOF and NIH in the quantum setting, we follow the work of Lee, Schechtman, and Shraibman \cite{Lee2009}, originally defined by Kerenidis \cite{Kerenidis2009}.
\begin{definition}[Quantum Multiparty Protocol]
Let $k$ be the number of players in the protocol. Define the Hilbert space by $\mathcal{H}_1\otimes\cdots \otimes \mathcal{H}_k\otimes \mathcal{C}$, where each $\mathcal{H}_i$ is the Hilbert space of player $i$, and $\mathcal{C}$ is the one-qubit channel. To communicate the players take turns predefined at the beginning of the protocol. On the turn of player $i$:
\begin{enumerate}
\item in NIH, an arbitrary unitary that only depends on $x_{i}$ is applied on $\mathcal{H}_i\otimes \mathcal{C}$, and acts as the identity anywhere else;
\item in NOF, an arbitrary unitary that depends on all inputs except $x_{i}$ is applied on $\mathcal{H}_i\otimes \mathcal{C}$, and acts as the identity anywhere else.
\end{enumerate}
The cost of the protocol is the number of rounds.
\end{definition}
The initial state is a pure state $\ket{0}\otimes\cdots\otimes\ket{0}\ket{0}$ without any prior entanglement. 
If the final state of the protocol on input $x_1,\dots,x_k$ is $\ket{\psi}$, it outputs 1 with probability $p(x_{1},\dots,x_{k})=\bra{\psi}\Pi_1\ket{\psi}$, where $\Pi_1$ is a projection onto the $\ket{1}$ state of the channel.

We say that $T$ is a {\it nondeterministic communication tensor} if $T[x_{1},\dots, x_{k}]\neq 0$ if and only if $f(x_{1},\dots,x_{k})=1$. Thus, $T$ can be obtained by replacing each 1-entry in the original communication tensor by a non-zero complex number. We also define the {\it nondeterministic rank} of $f$, denoted $nrank(f)$, to be the minimum rank over the complex field among all nondeterministic tensors for $f$.

\begin{definition}[Strong Quantum Nondeterministic Protocol]
A {\it $k$-party strong quantum nondeterministic communication protocol} outputs 1 with positive probability if and only if $f(x)=1$.
\end{definition}

The $k$-party quantum nondeterministic communication complexity, denoted $NQ_k(f)$, is the cost of an optimum (i.e., minimal cost) $k$-party quantum nondeterministic communication protocol. If the model is NIH or NOF, we add a superscript $NQ_k^{NIH}(f)$ or $NQ_k^{NOF}(f)$ respectively. From the definition it follows that $NQ_k$ is a lower bound for the exact quantum communication complexity $Q_k$ for both NOF and NIH.

The following lemma, given in Lee, Schechtman, and Shraibman \cite{Lee2009}, generalizes a previous observation made by Yao \cite{Yao1993} and Kremer \cite{Kremer1995} on 2-party protocols.


\begin{lemma}\label{lem:final-state}
After $\ell$ qubits of communication on input $(x_{1},\dots,x_{k})$, the state of a quantum protocol without prior shared entanglement can be written as
\[ \sum_{m\in\{0,1\}^\ell} \ket{A_m^{1}(x^1)}\ket{A_m^2(x^2)}\cdots \ket{A_m^k(x^k)}\ket{m_\ell}, \]
where $m_\ell$ is the $\ell$-th bit in $m$, and each vector $\ket{A_m^t(x^t)}$ corresponds to the $t$-th player which depends on $m$ and the input $x^t$. If the protocol is NOF then $x^t=(x_{1},\dots,x_{t-1},x_{t+1},\dots,x_{k})$; if it is NIH then $x^t=(x_{t})$.
\end{lemma}

\section{Proof of Theorem \ref{the:bounds}}\label{sec:bound}

\subsection{Lower Bound}
The arguments in this section are generalizations of a previous result by \cite{DeWolf2003} from 2-party  to $k$-party communication for $k\geq 3$. First we need the following technical lemma (see below for a proof).
\begin{lemma}\label{lem:family-vectors}
If there exist $k$ families of vectors such that $\{\ket{A_1^i(x_{i})},\dots,\ket{A_r^i(x_{i})}\}\subseteq \mathbb{C}^d$  for all $i$ with $1\leq i\leq k$ and $x_{i}\in\{0,1\}^n$ given that
\[
\sum_{i=1}^r \ket{A_i^1(x_{1})}\otimes\dots\otimes\ket{A_i^k(x_{k})}=0 \text{ iff } f(x_{1},\dots,x_{k})=0,
\]
then $nrank(f)\leq r$.
\end{lemma}

Now we proceed to prove the lower bound as stated in Theorem \ref{the:bounds}.
\begin{lemma}
$NQ_k^{NIH}(f)\geq \lceil \log nrank(f)\rceil+1$
\end{lemma}
\begin{proof}
Consider a NIH $\ell$-qubit protocol for $f$. By Lemma \ref{lem:final-state} its final state is
\begin{equation}
\ket{\psi}=\sum_{m\in \{0,1\}^\ell} \ket{A_m^1(x_{1})}\cdots \ket{A_m^k(x_{k})}\ket{m_\ell}.
\end{equation}
Assume all vectors have the same dimension $d$. Let $S=\{m\in\{0,1\}^\ell:m_\ell=1\}$, and consider only the part of the state that is projected onto the 1-state of the channel,
\begin{equation}
\ket{\phi(x_{1},\dots,x_{k})}=\sum_{m\in S} \ket{A_m^1(x_{1})}\cdots \ket{A_m^k(x_{k})}\ket{1}.
\end{equation}

The vector $\ket{\phi(x_{1},\dots,x_{k})}$ is 0 if and only if $f(x_{1},\dots,x_{k})=0$. Thus, by Lemma \ref{lem:family-vectors}, we have that $nrank(f)\leq |S|=2^{\ell-1}$, which implies the lower bound.
\end{proof}
\vspace{0.2cm}

\begin{proof}[Proof of Lemma \ref{lem:family-vectors}]
Let $k\geq 3$. We divide the proof in two cases, when $k$ is odd and even.
\vspace{0.2cm}

\noindent{\it Even $k$:} There are $k$ size-$r$ families of $d$-dimensional vectors. We will construct two new families of vectors denoted $\mathscr{D}$ and $\mathscr{F}$. First, divide the $k$ families in two groups of size $k/2$. Then, tensor each family in one group together in the following way: for each family $\{\ket{A_1^i(x_i)},\dots,\ket{A_r^i(x_i)}\}$ for $1\leq i \leq k/2$ construct a new family
\begin{align*}
\mathscr{D}	&=\left\{ \bigotimes_{j=1}^{k/2} \ket{A_1^j(x_{j})},\dots,\bigotimes_{j=1}^{k/2} \ket{A_r^j(x_{j})}\right\}\\
			&=\bigg\{  \ket{A_1(y)},\dots,\ket{A_r(y)} \bigg\},
\end{align*}
where $y=(x_{1},\dots,x_{k/2})$. Do the same to construct $\mathscr{F}$ for $k/2+1 \leq i \leq k$ obtaining
\begin{align*}
\mathscr{F}	&=\left\{ \bigotimes_{j=k/2+1}^{k} \ket{A_1^j(x_{j})},\dots,\bigotimes_{j=k/2+1}^{k} \ket{A_r^j(x_{j})}\right\}\\
			&=\bigg\{  \ket{B_1(z)},\dots,\ket{B_r(z)} \bigg\},\nonumber
\end{align*}
where $z=(x_{k/2+1},\dots,x_{k})$. Thus, $\mathscr{D}$ and  $\mathscr{F}$ will become two size-$r$ family of vectors, each vector with dimension $dk/2$. Then apply the theorem for $k=2$ from \cite{DeWolf2003} on these two families and the lemma follows.
\vspace{0.2cm}

\noindent{\it Odd $k$:} Here we can use the same approach by constructing again two new families $\mathscr{D}$ and $\mathscr{F}$ by dividing the families in two groups of size $\lfloor k/2\rfloor$ and $\lceil k/2\rceil$. However, although both families will have the same number of elements $r$, the dimension of the vectors will be different. In fact, the dimension of the vectors in one family will be $d'=d\lfloor k/2\rfloor$ and in the other $d'+1$. So, in order to prove the theorem we will consider having two families $\{\ket{A_1(y)},\dots,\ket{A_r(y)}\}\subseteq \mathbb{C}^{d'}$ and $\{\ket{B_1(z)},\dots,\ket{B_r(z)}\}\subseteq \mathbb{C}^{d'+1}$, both with cardinality $r$.

Denote the entry of each vector $\ket{A_i(y)},\ket{B_i(z)}$ by $A_i(y)_u$ and $B_i(z)_v$ respectively for all $(u,v)\in[d']\times[d'+1]$. Note that, if $f(y,z)=0$ then $\sum_{i=1}^r A_i(y)_u B_i(z)_v=0$ for all $(u,v)$; if $f(y,z)=1$ then $\sum_{i=1}^r A_i(y)_u B_i(z)_v\neq 0$ for some $(u,v)$. This holds because each vector $\ket{A_i(y)}$ and $\ket{B_i(z)}$ are the set of vectors $\ket{A_i^t(x^{t})}$ tensored together and separated in two families of size $\lfloor k/2\rfloor$ and $\lceil k/2\rceil$  respectively.

The following lemma was implicitly proved by de Wolf \cite{DeWolf2003} for families of vectors with the same dimension. However, we show that the same arguments hold even if the families have different dimensionality (see \ref{app:lemmas} for a proof).
\begin{lemma}\label{lem:existence}
Let $I$ be an arbitrary set of real numbers of size $2^{2n+1}$. Let $\alpha_1,\dots,\alpha_{d'}$ and $\beta_1,\dots,\beta_{d'+1}$ be numbers from $I$, and define the quantities
\[
a_i(y)=\sum_{u=1}^{d'} \alpha_u A_i(y)_u \quad\text{and}\quad b_i(z)=\sum_{v=1}^{d'+1} \beta_v B_i(z)_v.
\]
Also let
\[
v(y,z)=\sum_{i=1}^r a_i(y)b_i(z)=\sum_{u=1}^{d'} \sum_{v=1}^{d'+1}\alpha_u \beta_v \left(\sum_{i=1}^r A_i(y)_u B_i(z)_v\right).
\]
There exists $\alpha_1,\dots,\alpha_{d'},\beta_1,\dots,\beta_{d'+1}\in I$ such that for every $(y,z)\in f^{-1}(1)$ we have $v(y,z)\neq 0$. 
\end{lemma}

Therefore, by the lemma above we have that $v(y,z)=0$ if and only if $f(y,z)=0$. Now let $\ket{a_i}$ and $\ket{b_i}$ be $2^n$-dimensional vectors indexed by elements from $\{0,1\}^n$, and let $M=\sum_{i=1}^r \ket{a_i}\bra{b_i}$. Thus $M$ is a nondeterministic order-$k$ tensor of rank $r$.
\end{proof}

\subsection{Upper Bound}
The proof of the upper bound follows by fixing a proper matrization (separating the cases of odd and even $k$) of the communication tensor, and then applying the 2-party protocol by de Wolf\cite{DeWolf2003}.

\begin{lemma}\label{lem:nof-protocol}
$NQ_k^{NOF}(f)\leq\lceil \log nrank(f)\rceil+1$.
\end{lemma}

\begin{proof}
Let $T$ be a nondeterministic tensor for $f$ with $nrank(f)=r$. We divide the proof in two cases.
\vspace{0.2cm}

\noindent {\it Even $k$:} Fix two players, say $P_1$ (Alice) and $P_k$ (Bob).  Also fix some matrization of $T$, i.e., let $M$ be such matrization and consider it as an operator $M:\mathcal{H}_{k/2+1}\otimes\cdots\otimes \mathcal{H}_k\to\mathcal{H}_1\otimes\cdots\otimes\mathcal{H}_{k/2}$. Thus $M$ is a $2^{kn/2}\times 2^{kn/2}$-matrix that maps elements from the $\mathcal{H}_{k/2+1}\otimes\cdots\otimes \mathcal{H}_k$ subspace to the $\mathcal{H}_1\otimes\cdots\otimes\mathcal{H}_{k/2}$ subspace. Let also $M=U\Sigma V$ be the singular value decomposition of $M$ such that $U,V$ are $2^{kn/2}\times 2^{kn/2}$ unitary matrices, and $\Sigma$ is a $2^{kn/2}\times 2^{kn/2}$ diagonal matrix containing the singular values of $M$ in the diagonal. The number of singular values is at most $rank(M)\leq r$.

Bob computes the state $\ket{\phi_{1\cdots k/2}}=c_{1\cdots k/2} \Sigma V\ket{x_1,\dots,x_{k/2}}$ where $c_{1\cdots k/2}$ is some normalizing constant that depends on $x_1,\dots,x_{k/2}$. Since only the first entries of $\Sigma$ are non-zero, $\ket{\phi_{1\cdots k/2}}$ has at most $r$ non-zero entries, so the state can be compressed using $\log r$ qubits\footnote{A $n$ dimensional vector can be encoded as a quantum state with $\log n$ qubits by observing that a $k$-qubit state is a $2^k$-dimensional vector. This fact was used by Raz \cite{Raz1999} to show an exponential separation between classical and quantum 2-party communication.}. Bob sends these qubits to Alice. Alice then computes $U\ket{\phi_{1\cdots k/2}}$ and measures that state. If Alice observes $x_{k/2+1},\dots,x_k$ then she puts a 1 on the qubit channel, and otherwise she puts a 0. The probability of Alice putting a 1 on the channel is
\begin{IEEEeqnarray*}{rCl}
\IEEEeqnarraymulticol{3}{l}{\left|\bra{x_{k/2+1},\dots,x_k}U\ket{\phi_{1\cdots k/2}}\right|^2}\\
	\quad\quad&=& |c_{1\dots,k/2}|^2\left|\bra{x_{k/2+1},\dots,x_k}U\Sigma V\ket{x_1,\dots,x_{k/2}}\right|^2\\
	\quad\quad&=&|c_{1\dots,k/2}|^2\left|\bra{x_{k/2+1},\dots,x_k}M\ket{x_1,\dots,x_{k/2}}\right|^2\\
	\quad\quad&=&|c_{1\dots,k/2}|^2 \left|M[x_1,\dots, x_k]\right|^2\\
	\quad\quad&=&|c_{1\dots,k/2}|^2 \left|T[x_1,\dots, x_k]\right|^2.
\end{IEEEeqnarray*}
Since $T[x_1,\dots, x_k]$ is non-zero if and only if $f(x_1,\dots,x_k)=1$, this probability will be positive if and only if $f(x_1,\dots,x_k)=1$. Thus, this is a nondeterministic protocol with total cost  $\log r+1$.
\vspace{0.2cm}

\noindent{\it Odd $k$:} To use the protocol given in the even case, we add an extra degree of freedom to $T$.

\begin{lemma}\label{lem:tensor}
If $T$ is an order-$k$ tensor with rank $r$ then, there exists a tensor $T'$ of order $k+1$ with rank $r$ where $T[x_1,\dots, x_k]=T'[x_1,\dots, x_k x_{k+1}]$ for all $x_{k+1}$.
\end{lemma}

By the above lemma we  have that $T'[x_1,\dots, x_k x_{k+1}]=0$ if and only if $f(x_1,\dots,x_k)=0$ for any given $x_{k+1}$. See \ref{app:lemmas} for a proof.

Before the protocol starts, each player knows $T'$ (which has even order) and its matrization $M'$. We fix two players, $P_1$ (Alice) and $P_k$ (Bob), and they can now use the protocol for even $k$.
\end{proof}

\section{Rank Lower Bound for the Generalized Inner Product}\label{app:gip-rank}
In this section we give a lower bound on the nondeterministic rank of the Generalized Inner Product (GIP) function.

\begin{lemma}
$nrank(GIP_k)\geq (k-1)2^{n-1}+1$.
\end{lemma}
\begin{proof}
First, we start by generalizing the concept of rows and columns for tensors. Define a {\it fiber} to be a vector obtained by fixing every index except by one. In general, a mode-$i$ fiber is a vector obtained by fixing all except the $i^{th}$ index. Thus, a matrix column is a mode-1 fiber, and a row is a mode-2 fiber. For order-3 tensors, we have columns, rows and tubes, and so on for higher order tensors.  In the same way we define a {\it slice} to be a two-dimensional section of $T$ obtained by fixing all but two indices.

Here we will consider a particular form of matrization. Let $T\in \mathbb{C}^{n_1\times \cdots \times n_k}$ be an order-$k$ tensor, with $n_i=2^n$ for every $i$. The {\it $i$-mode unfolding} of $T$, denoted $T_{(i)}$,  is the matrix obtained by arranging the $i$-mode fibers as columns. The permutations of the columns of $T_{(i)}$ is not important, as long as the corresponding operations remain consistent; see Kolda and Bader \cite{Kolda2009}. Define the $i$-$rank$ of $T$ as $rank_i(T)=rank(T_{(i)})$. It is trivial that $rank_i(T)\leq rank(T)$ for every $i$; see Lathauwer, de Moore, and Vandewalle \cite{Lathauwer2000}.

Now we proceed with the proof. Let $T$ be the order-$k$ nondeterministic communication tensor for $GIP_k$. Let $M_{IP_n}$ be the Boolean communication matrix for $GIP_2$, i.e., the 2-party inner product function on $n$ bits. It is well known that $rank(M_{IP_n})=2^n-1$; see Example 1.29 in Kushilevitz and Nisan\cite{Kushilevitz1997}. The same holds even if $M_{IP_n}$ is defined over $\mathbb{C}$.

Let ${\bf 1}$ denote the string of length $n$ with only 1s in it, and let $T'$ be the $(x_3',\dots,x_k')$-slice of $T$ where $x_i'={\bf 1}$ for $i=3,\dots,k$. In this way $T'[x_1,x_2]\neq 0$ whenever $\braket{x_1}{x_2}=1$ and hence $rank(T')=rank(M_{IP_n})=2^n-1$.

Let $x^{(i)}$ denote the string $x$ with the $i^{th}$ bit flipped. For $i=3,\dots,k$ consider the $(x_3',\dots,x_k'^{(i)})$-slice of $T$ denoted $T_i'$ where $x_k'^{(i)}$ is the string ${\bf 1}$ with the $i^{th}$ bit flipped to 0. Then,
\begin{equation}\label{eq:matrization}
T_i'[x_1,x_2]\neq 0\text{ whenever } \braket{x_1}{x_2}-x_{1i}x_{2i}=1.
\end{equation}
Note that the non-zero entries of $T_i'$ for any $i$ agrees with the non-zero entries of $M_{IP_{n-1}}$, where $M_{IP_{n-1}}$ is obtained by deleting the $i^{th}$ bits of $x_1$ and $x_2$ in $M_{IP_n}$ for all $x_1$ and $x_2$. Thus, $rank(T_i')=2^{n-1}-1$ for all $i=3,\dots,k$.

The $1$-mode unfolding of $T$ is obtained by fixing every index except $x_1$. Thus
\[
T_{(1)}=\begin{bmatrix}
T'	& T_3'	& \cdots	& T_k'	& \cdots
\end{bmatrix},
\]
with $2^{(k-1)n}$ columns, and the right part of $T_{(1)}$ (after $T_k'$) is filled with the remaining slices of $T$ that are different to $T'$ and each $T_i'$.  We known that $T'$ and each $T_i'$ have $(2^n-1)$ and $2^{n-}-1$ linearly independent columns respectively. Also, each of these columns are pair-wise linearly independent. To see this, just take take any two slices $T_i'$ and $T_j'$ for any $i\neq j$, fix one column in each and compute the inner product according to Equation \ref{eq:matrization}. Thus, $rank(T)\geq rank_1(T)\geq 2^n-1+(k-2)(2^{n-1}-1)=(k-1)2^{n-1}+1$.
\end{proof}

\section{Some Separations for Complexity Classes}\label{sec:separation}
In this section we take a complexity-theoretic view of quantum multiparty communication complexity. For this model we consider as ``efficient communication'' when a protocol computes a function with $polylog(n)$ bits \cite{Babai1986}.

\begin{definition}
We define the following communication complexity classes:
\begin{enumerate}
\item $BPP^{cc}$ is the class of boolean functions with a classical bounded-error protocol of cost $polylog(n)$;
\item $BQP^{cc}$ is the class of boolean functions with a quantum bounded-error protocol of cost $polylog(n)$;
\item $NQP^{cc}$ is the class of boolean functions with a quantum strong nondeterministic protocol of cost $polylog(n)$.
\end{enumerate}
\end{definition}

In the following we present two theorems that give separations between the complexity classes defined above. First, for better understanding, we start by showing a weaker nevertheless easier to prove result, a separation between $NQP^{cc}$ and $BPP^{cc}$. Then we use that result to separate $NQP^{cc}$ from $BQP^{cc}$. Although this latter result can be proved without the need of the former, starting with the separation from $BPP^{cc}$ seems easier to understand.

\begin{theorem}\label{the:bpp-gap}
For NOF communication we have that $NQP^{cc} \nsubseteq BPP^{cc}$ whenever the number of players $k=o(\log\log n)$.
\end{theorem}
\begin{proof}
To prove this we exhibit a function $f:X^k\to \{0,1\}$ such that $NQ_k^{NOF}(f)=O(\log n)$ and $R_{\epsilon,k}(f)=\Omega(n^{1/(k+1)}/(k2^{2^{k}}))$, where $R_\epsilon$ denotes the bounded-error NOF communication complexity with error probability upper-bounded by $\epsilon$. This will give the separation whenever $k=o(\log\log n)$.

In particular, we analyze the following total function. Let $x_1,\dots,x_k\in X$ with $X=\{0,1\}^n$, then
\begin{equation}\label{eq:function}
f(x_1,\dots,x_k)=\left\{\begin{array}{cl}
1	&\text{if }|x_1\land \cdots \land x_k|\neq 1\\
0	&\text{if }|x_1\land \cdots \land x_k|= 1
\end{array}\right.,
\end{equation}
where $\land$ denotes the bit-wise AND and $|x|$ is the Hamming weight of $x$. This function was previously studied by de Wolf \cite{DeWolf2003} in the 2-player case.
\vspace{0.2cm}

\noindent\emph{Upper Bound:} For each $i$ let $x_i=x_{ij_1}\dots x_{ij_n}$ and let $T_j$ be an order-$k$ tensor where $T_j[x_1,\dots,x_k]=1$ if $x_{1j}\land \cdots \land x_{kj}=1$ and $T_j[x_1,\dots,x_k]=0$ otherwise. Note that for each $j$ the tensor $T_j$ has rank 1. Define the order-$k$ tensor $T$ by
\[
T[x_1,\dots,x_k]=\sum_{j=1}^n T_j[x_1,\dots,x_k]-1.
\]
This tensor has rank $n$. Also $T$ is a nondeterministic communication tensor for $f$ since $T[x_1,\dots,x_k]=0$ if and only if $|x_1\land \cdots \land x_k|=1$. Hence, by Theorem \ref{the:bounds} the upper bound follows.
\vspace{0.2cm}

\noindent\emph{Lower Bound:}
To prove the lower bound we will use, without loss of generality, the sign version of Equation (\ref{eq:function}), i.e., 
\begin{equation}\label{eq:sign-function}
f(x_1,\dots,x_k)=\left\{\begin{array}{cl}
1	&\text{if }|x_1\land \cdots \land x_k|\neq 1\\
-1	&\text{if }|x_1\land \cdots \land x_k|= 1
\end{array}\right..
\end{equation}

We make use of a result by Lee and Shraibman \cite{Lee2008}. Let $\mu^\alpha$ be the \emph{approximate cylinder intersection norm} as defined in \cite{Lee2008}, and let $\widetilde{deg}(f)$ be the \emph{approximate degree} of a boolean function $f$ \cite{Nisan1992}.





\begin{lemma}\label{lem:deg-bound}
Let $f_n:\{0,1\}^n\to\{-1,1\}$ be a symmetric\footnote{A function is called symmetric if it only depends on the number of 1s in the input.} function, and let $F_f:(\{0,1\}^n)^k\to\{-1,1\}$ be a function (not necessarily symmetric) defined by $F_f(x_1,\dots,x_k)=f(x_1\land \dots \land x_k)$. Let $\alpha>1/(1-2\epsilon)$ and set $c=2e(k-1)2^{2^{k-1}}$, then
\[ R_{1/4,k}(F_{f_n})=\Omega(\log \mu^\alpha(F_{f_n}))=\Omega\left(\frac{\widetilde{deg}(f_m)}{2^{k}}\right), \]
where $n=(c/\widetilde{deg}(f_m))^{k-1}m^k$.
\end{lemma}

Note that Lemma \ref{lem:deg-bound} is a generalization of \cite[Corollary 6.1]{Lee2008} to symmetric functions. However, as pointed by the authors of \cite{Lee2008}, this generalization is straightforward and can be easily proved by following the proof of \cite[Corollary 6.1]{Lee2008}, and it is therefore omitted from this paper.

Define the following Hamming weight function:
\[
w(x)=\left\{\begin{array}{cl}
1	&\text{if } |x|\neq 1\\
-1	&\text{if } |x|=1
\end{array}\right..
\]
This way we can write Equation \ref{eq:sign-function} as $f(x_1,\dots,x_k)=w(x_1\land \dots \land x_k)$. Also note that $w$ is symmetric and we can apply Lemma \ref{lem:deg-bound}. Together with the characterization given by Paturi \cite{Paturi1992} of the approximate degree of symmetric functions we have that
\begin{equation}\label{eq:lower-bound}
\log \mu^\alpha(f)=\Omega\left(\frac{n^{1/(k+1)}}{k2^{2^{k}}} \right).
\end{equation}
\end{proof}

\begin{theorem}
For NOF communication we have that $NQP^{cc} \nsubseteq BQP^{cc}$ whenever the number of players $k=o(\log\log n)$.
\end{theorem}
\begin{proof}
 To prove this we rely again in Equation (\ref{eq:sign-function}) and the fact that $NQ_k^{NOF}(f)=O(\log n)$. Here we show that $Q_{\epsilon,k}(f)=\Omega(n^{1/(k+1)}/(k2^{2^{k}})-k)$,  where $Q_\epsilon$ denotes the bounded-error NOF communication complexity with error probability upper-bounded by $\epsilon$.

Note that to prove Theorem \ref{the:bpp-gap} we derived a lower bound on $\mu^\alpha$. We can use the same lower bound to prove the separation for $BQP^{cc}$. In order to do that we make use of the following two results by Lee, Schechtman, and Shraibman \cite{Lee2009}. Let $\gamma^\alpha$ be the \emph{approximate quantum norm} as defined in \cite{Lee2009}.
\begin{lemma}
Let $T$ be an order-$k$ sign-tensor, then $Q_{\epsilon,k}(T)=\Omega(\log \gamma^\alpha(T))$.
\end{lemma}
\begin{lemma}
For every order-$k$ tensor $T$, $\gamma(T) \leq \mu(T) \leq C^k \gamma(T)$, for some absolute constant $C$.
\end{lemma}

Thus, by these two lemmas above and Equation \ref{eq:lower-bound} we have that
\[
\log \gamma^\alpha(f)=\Omega\left(\frac{n^{1/(k+1)}}{k2^{2^{k}}} -k\right).
\]
\end{proof}

\section{Concluding Remarks}
In this paper we studied strong quantum nondeterministic communication complexity in multiparty protocols. In particular, we showed that i) strong quantum nondeterministic NOF communication complexity is upper-bounded by the logarithm of the rank of the nondeterministic communication tensor; ii) strong quantum nondeterministic NIH communication complexity is lower-bounded by the logarithm of the rank of the nondeterministic communication tensor. These results naturally generalizes previous work by de Wolf\cite{DeWolf2003}. Moreover, the lower bound on NIH is also a lower bound for quantum exact NIH communication. This fact was used to show a $\Omega(n+\log k)$ lower bound for the generalized inner product function.

We also showed that $NQP^{cc}\nsubseteq BPP^{cc}$ and $NQP^{cc}\nsubseteq BQP^{cc}$ when the number of players is $o(\log\log n)$. It remains as an open problem to prove the same separations with an increased number of players.

In order to prove strong lower bounds using tensor-rank in NIH, we need stronger construction techniques for tensors. The fact that computing tensor-rank is $NP$-complete suggests that this could be a very difficult task. Alternatives for finding lower bounds on tensor-rank include computing the norm of the communication tensor, or a hardness result for approximating tensor-rank.

\section*{Acknowledgements}
The authors thank the anonymous reviewers from TAMC'12 for initial reviews of this paper. The first author thanks the NEC C\&C Foundation for partially supporting this research.

\bibliographystyle{alpha}
\bibliography{../../library}

\appendix

\section{Proofs of Technical Lemmas}\label{app:lemmas}

\subsection{Proof of Lemma \ref{lem:existence}}

If $f(y,z)=0$ then $v(y,z)=$ for all $\alpha_u,\beta_v$. If $f(y,z)\neq 0$ there exists $(u',v')$ such that $v(y,z)\neq 0$. Here we use the same arguments given by \cite{DeWolf2003}, i.e., we show that $v(y,z)=0$ happens with small probability. In fact, having families of vectors with different dimensions does not affect the argument. Consider the situation where all $\alpha_u$ and $\beta_v$ were chosen except $\alpha_{u'}$ and $\beta_{v'}$. Write $v(y,z)$ in terms of these two coefficients
\[
v(y,z)=c_0\alpha_{u'}\beta_{v'}+c_1\alpha_{u'}+c_2\beta_{v'}+c_3,
\]
where $c_0=\sum_{i=1}^r A_i(y)_{u'} B_i(z)_{v'}\neq 0$. If we fix $\alpha_{u'}$ then, $v(y,z)$ is a linear equation with at most one zero for each $\alpha_{u'}$. Therefore, we have at most $2^{2n+1}+2^{2n+1}-1=2^{2n+2}-1$ ways of choosing $\alpha_{u'}$ and $\beta_{v'}$ such that $v(y,z)=0$. Thus
\[
Pr[v(y,z)=0]\leq \frac{2^{2n+1}}{(2^{2n+1})^2}< \frac{ 2^{2n+2}}{(2^{2n+1})^2}=2^{-2n}.
\]
By the union bound
\begin{IEEEeqnarray*}{rCl}
\IEEEeqnarraymulticol{3}{l}{Pr[\exists(y,z)\in f^{-1}(1) \text{ s.t. } v(y,z)=0]}\\
 \quad\quad\quad&\leq& \sum_{(y,z)\in f^{-1}(1)} Pr[v(y,z)=0]<2^{2n}\cdot 2^{-2n}=1.
\end{IEEEeqnarray*}
The following is a probabilistic method argument. Since the above probability is strictly less than 1, there exists sets $\{a_1(y),\dots,a_r(y)\}$ and $\{b_1(z),\dots,b_r(z)\}$ such that for every $(y,z)\in f^{-1}(1)$ we have $v(y,z)\neq0$.

\subsection{Proof of Lemma \ref{lem:tensor}}\label{sec:tensor}
Let $T=\sum_{i=1}^r \ket{v_1^i} \cdots \ket{v_k^i}$ for some family of $d$-dimensional vectors.  Define the tensor $T'=\sum_{i=1}^{r} \ket{v_1^i} \cdots  \ket{v_k^i} \ket{v_{k+1}^i}$ where each $\ket{v_{k+1}^i}$ is the all-1 vector. Thus, component-wise we have that
\[T[x_1,\dots, x_k]=\sum_{i=1}^r v_1^i(x_1)\cdots v_k^i(x_k),\]
and
\[T'[x_1,\dots, x_kx_{k+1}]=\sum_{i=1}^{r} v_1^i(x_1)\cdots v_k^i(x_k)v_{k+1}^i(x_{k+1}),\]
where $v_{k+1}^i(x_{k+1})=1$ for all $i$ and for all inputs $x_{k+1}$. Then $T'[x_1,\dots, x_kx_{k+1}]=\sum_{i=1}^{r} v_1^i(x_1)\cdots v_k^i(x_k)$ and $T'[x_1,\dots, x_k x_{k+1}]=T[x_1,\dots, x_k]$ for any $x_{k+1}$.

\end{document}